\documentclass[12pt]{article}

\usepackage{cite}
\usepackage{amsmath,amssymb}

\newcommand{\be}{\begin{equation}}
\newcommand{\en}{\end{equation}}

\newcommand{\F}{\mathcal{F}}

\newcommand{\Lc}{\mathcal{L}}

\newcommand{\1}{1 \!\! 1}

\newcommand{\Hil}{\mathcal{H}}

\newtheorem{thm}{Theorem}
\newtheorem{remark}{Remark}
\newtheorem{example}{Example}
\newenvironment{proof}{\noindent {\bf Proof.}}{\hfill$\square$ \vspace{3mm}\endtrivlist}

\begin{document}

\title{Some classes of finite-dimensional \\ladder operators\footnote{To appear on Journal of Physics A: Mathematical and Theoretical, 2025.}}

\author{Fabio Bagarello\footnote{Dipartimento di Ingegneria, Universit\`a di Palermo,
Viale delle Scienze, Edificio 8, I--90128  Palermo, Italy; 
I.N.F.N., Sezione di Catania, Italy. fabio.bagarello@unipa.it}, 
Antonino Faddetta\footnote{Dipartimento di Fisica e Chimica, Universit\`a di Palermo, Via Archirafi 36, I--90128  Palermo, Italy. antonino.faddetta@community.unipa.it}, 
Francesco Oliveri\footnote{Dipartimento di Scienze Matematiche e Informatiche, Scienze Fisiche e Scienze della Terra,  Universit\`a di Messina, Viale F. Stagno d'Alcontres 31, I--98166 Messina, Italy. francesco.oliveri@unime.it}}

\maketitle

\begin{abstract}
We introduce and study some special classes of ladder operators in finite-dimensional Hilbert spaces. In particular we consider a truncated version of quons, their {\em psudo-}version, and a third family of operators acting on a closed chain. In this latter situation, we discuss the existence of what could be considered {\em discrete coherent states}, as suitable eigenvectors of the annihilation operator of the chain. We see that, under reasonable assumptions, a resolution of the identity can be recovered, involving these states, together with a biorthogonal family of vectors, which turn out to be eigenstates of the raising operator of the chain.
\end{abstract}

\noindent{\it Keywords}: quons, ladder operators, closed chain, coherent states.

\noindent {\it MSC classification}: 46N50, 81Q12.

\section{Introduction}\label{sect1}

The role of ladder operators is very useful in many fields of physics, from quantum mechanics to quantum field theory, passing through many body problem and complex systems. There are hundreds of books dealing with this class of operators, both in view of their applications but also because they produce a beautiful mathematics. We cite here only few monographs on this topics \cite{mer,mess,cohen,gazeaubook}, but many others can be easily found, with focus on different aspects of ladder operators, different extensions, and different realms of applications. 

Bosons and fermions are the two classes of ladders which are likely the most known since the beginning of quantum mechanics, the reason being that they describe interactions between particles (the bosons), and the particles themselves (the fermions). We could think of photons and  electrons, respectively, just to give two concrete examples of  bosons and fermions, respectively. Nevertheless,  there are many other classes of ladder operators in the literature. Quons are one such class \cite{moh,gre}. Abstract ladder operators \cite{fern} produce another class. And there are more: those related to the angular momentum \cite{mer,mess,cohen}, or those arising from the generalized Heisenberg algebra \cite{curado}. And more. Some of them live in finite-dimensional, others in infinite-dimensional Hilbert spaces. This aspect is related to the explicit definition of the operators. Moreover, some operators originally defined in infinite-dimensional Hilbert spaces admit a counterpart in finite-dimensional ones. This {\em restriction} can be motivated by physical reasons: for instance, bosons are used in the diagonalization of the quantum harmonic oscillator. Its Hamiltonian, $H_{qo}$, admits an infinite set of (equally spaced) eigenvalues and an infinite set of mutually orthogonal eigenvectors. Anyway, it is clear that 
in a realistic situation not all the energetic levels can be reached by the oscillator. Only a (possibly large) finite number of levels can be occupied and, therefore, the \emph{original} Hilbert space $\Lc^2(\mathbb{R})$ is \emph{too big}! It is sufficient to consider the subspace of $\Lc^2(\mathbb{R})$ which is generated by a finite number of eigenvectors of $H_{qo}$, say those corresponding to the energetic levels which can really be occupied. This kind of restriction is discussed, for instance, in \cite{buc,bagchi,bag2013}, where {\em truncated bosons} have been analyzed. We should also mention that all the ladder operators named so far can be extended or deformed in such a way that the raising operator is {\bf not} necessarily the adjoint of the lowering operator, as it often happens in quantum mechanical systems. A recent overview of many of these extensions is contained in \cite{bagbook}.

The main focus in this paper is to consider families of ladders which live in finite-dimensional Hilbert spaces different from those considered so far in the existing literature, and to deduce their main properties. In particular, we will concentrate on a finite-dimensional version of quons, and on their pseudo-deformations. Then, we will discuss also ladder operators on a closed chain. More in detail, the paper is organized as follows. In Section~\ref{sect2}, we will introduce the {\em truncated quons} (TQs), whereas in Subsection~\ref{sect3} we will deform these TQs, so getting truncated pseudo-quons (TPQs) similar to those considered for infinite-dimensional Hilbert spaces in \cite{bagbook}. Section~\ref{sect4} is devoted to the analysis of ladders on a closed chain, with a special focus on the possibility of defining some sort of coherent states for the chain. Section~\ref{sect5} contains our conclusions and plans for the future. To keep the paper self-contained, we list some results on quons in the Appendix.

\section{From bosons to TQs}
\label{sect2}

Our main starting ingredient is an operator $a$ satisfying, together with its adjoint $a^\dagger$, the  canonical commutation relation (CCR) $[a,a^\dagger]=\1$. Here, $\1$ is the identity operator on $\Hil$, the Hilbert space where $a$ and $a^\dagger$ act. It is well known that (i) $\Hil$ is necessarily infinite-dimensional and that (ii) $a$ and $a^\dagger$ are unbounded. An important consequence of the CCR is that, calling $N_a=a^\dagger a$, then $[N_a,a]=-a$. This is one of the key identities needed to use $a$ and $a^\dagger$ to diagonalize, in a purely algebraic way, the Hamiltonian of a harmonic oscillator.

In \cite{buc,bagchi}, it has been shown that it is possible to construct a sort of truncated version of the operators $a, a^\dagger$ which live in finite dimensional Hilbert spaces. For that, a possible approach consists in replacing the CCR with a different commutation rule. Indeed, let us consider an operator $A$ which satisfies, together with $A^\dagger$, the commutation relation
\be
[A,A^\dagger]=\1_L-(L+1)K_0,
\label{21}
\en
where $0<L<\infty$, $\1_L$ is the $(L+1)\times(L+1)$ identity matrix, and $K_0$ is an orthogonal projector ($K_0=K_0^2=K_0^\dagger$) satisfying $K_0A=0$. In this case it is possible to prove that $A$ admits a representation as a matrix of order $(L+1)$. In other words, $A$ is an operator acting on $\Hil_L$, where $\dim(\Hil_L)=L+1$, and $A^\dagger$ is the adjoint of $A$ under the usual scalar product in $\Hil_L$. We stress once more that this finite-dimensional representation is not possible if the {\em correction} of the CCR proportional to $K_0$ is missing in (\ref{21}).

To be more explicit, we recall that the original bosonic operator $a$ satisfies the lowering property on $\Hil$
\be 
ae_{n}=\sqrt{n}\,e_{n-1}, \, n\geq1, \qquad \mbox{ where  }\qquad  ae_{0}=0,
\label{22}
\en
with its standard representation
\be
a=\sum_{n=0}^\infty\sqrt{n+1}\,|e_{n}\rangle \langle e_{n+1}|.
\label{23}
\en
This series converges on $\Lc_e=l.s\{e_n\}$, the linear span of the vectors $e_n$, where 
$\F_e=\{e_n, \,n\geq0\}$ is the orthonormal (o.n.) basis of $\Hil$ constructed in the standard way: $e_0$ is a vector in $\Hil$ satisfying $ae_0=0$, and $e_n=\frac{1}{\sqrt{n!}}\,(a^\dagger)^ne_0$, $n\geq1$. It is well known that,  computing $[a,a^\dagger]$ (in the sense of unbounded operators), we get the identity operator on $\Hil$. Indeed, using a formal approach (which, however, could be made rigorous), we find that
$$
[a,a^\dagger]=\sum_{n=0}^\infty\,|e_{n}\rangle \langle e_{n}|=\1,
$$
since $\F_e$ is an o.n basis for $\Hil$, and therefore resolves the identity. An interesting feature of $a$ is that it admits, among the others, the following (infinite-dimensional) matrix representation:
$$
a=\left(
\begin{array}{cccccc}
0 & 1 & 0 & 0 & 0 & \ldots \\
0 & 0 & \sqrt{2} & 0 & 0 & \ldots \\
0 & 0 & 0 & \sqrt{3} & 0 & \ldots \\
0 & 0 & 0 & 0 & \sqrt{4} & \ldots \\
\vdots & \vdots & \vdots & \vdots & \vdots & \vdots 
\end{array}
\right).
$$

Let us now consider a sort of cut-off on $a$, replacing the infinite series in (\ref{23}) with the finite sum
\be
A=\sum_{n=0}^{L-1}\sqrt{n+1}\,|e_{n}\rangle \langle e_{n+1}|, 
\label{24}
\en
so that  $\displaystyle A^\dagger=\sum_{n=0}^{L-1}\sqrt{n+1}\,|e_{n+1}\rangle \langle e_{n}|$, where $e_n\in\F_e$. It is possible to check that (\ref{24}) produces an explicit realization of (\ref{21}), and that $A$, $A^\dagger$, and other relevant operators, can all be represented in a Hilbert space $\Hil_L\subset\Hil$ for which $\F_e(L)=\{e_n, n=0,1,2,\ldots,L\}$, finite subset of $\F_e$, is an o.n. basis. $\Hil_L$ can be seen as the set of all the vectors $f$ of $\Hil$ for which, when expanded in terms of $\F_e$, we get $\displaystyle f=\sum_{n=0}^\infty c_ne_n$ with $c_n=0$ for all $n> L$. Now we can easily check that
$$
[A,A^\dagger]=\sum_{n=0}^{L}|e_{n}\rangle \langle e_{n}|-(L+1)|e_{L}\rangle \langle e_{L}|,
$$
which is exactly formula (\ref{21}) with $\displaystyle \sum_{n=0}^{L}|e_{n}\rangle \langle e_{n}|=\1_L$, the identity in $\Hil_L$ (with dimension $L+1$), and $K_0=|e_{L}\rangle \langle e_{L}|$. Notice that, as required, $K_0=K_0^\dagger=K_0^2$. The matrix expression of $A$ is
$$
A=\left(
\begin{array}{cccccccc}
	0 & 1 & 0 & 0 & 0 & \cdots & \cdots & 0 \\
	0 & 0 & \sqrt{2} & 0 & 0 & \cdots & \cdots & 0 \\
	0 & 0 & 0 & \sqrt{3} & 0 & \cdots & \cdots & 0 \\
	0 & 0 & 0 & 0 & \sqrt{4} & \cdots & \cdots & 0 \\
	\vdots & \vdots & \vdots & \vdots & \vdots & \vdots & \vdots &0 \\
	0 & 0 & 0 & 0 & 0 & 0 & \cdots & \sqrt{L} \\
	0 & 0 & 0 & 0 & 0 & 0 & \cdots & 0 \\
\end{array}
\right),
$$
which looks as a truncated version of the matrix expression for $a$.

Hereafter, what we are interested to is the possibility of repeating a similar construction for quons (see Appendix), in particular to understand if the approach described here for bosons is the only one, or if some similar, but possibly slightly different, technique can also be proposed.

More explicitly, adopting the same notation as in Appendix, our main effort here is to understand if the operator $c$ satisfying, as in (\ref{a1}), $[c,c^\dagger]_q:=c c^\dagger -q c^\dagger c=\1$, $q\in[-1,1]$, can be somehow {\em replaced} by a different operator, $C$, satisfying a sort of generalized version of (\ref{21}): 
\be
[C,C^\dagger]_q=\1_L-(L+1)K,
\label{25}\en
where $0<L<\infty$ is again connected to the dimensionality of a suitable Hilbert space where the operators $C$ and $C^\dagger$ can be represented as $(L+1)\times (L+1)$ matrices, and $K$ is some operator required to get this finite-dimensional representation, as for the bosonic case. 

\begin{remark} 
\label{remark1}
It is useful to prove that $[c,c^\dagger]_q=\1$ can be defined on some finite-dimensional Hilbert space only in few  very special cases. Indeed, if we assume that $[c,c^\dagger]_q=\1$ can be implemented in some $\Hil_L$, with dimension $L+1<\infty$,  then taking the trace of both sides of the identity $[c,c^\dagger]_q=\1$, and using the properties of the trace, we would get $(1-q)\,tr(c^\dagger c)=L+1$. Now, since $c^\dagger c$ is a positive operator, its trace is non negative, and it is strictly positive if we exclude the trivial case that $c^\dagger c=0$. Hence, $q=1$ is not compatible with this result (in agreement with what we know for bosons). Now, if we assume first that $c$ is described by a truncated $(L+1)\times(L+1)$ version of (\ref{a2}) with $\beta_n$ given in (\ref{a5}), we see that $\displaystyle\hbox{tr}(c^\dagger c)=\sum_{n=0}^{L-1}\beta_{n}^2$, and therefore, using (\ref{a5}), we should have
$$
\sum_{n=0}^{L-1}(1-q^{n+1})=L+1,
$$
which is in general not true for our values of $q$. The only cases where this equality can be satisfied is when $L$ is odd and $q=-1$. For $L=1$, we have the fermionic case, and our conclusion is, therefore, expected.

There is also another solution of $[c,c^\dagger]_q=\1$: in fact, if $q\in[-1,1[$, the $(L+1)\times(L+1)$ matrix
$$c=\frac{1}{\sqrt{1-q}}
\left(
\begin{array}{cccccc}
    0 & 1 & 0 & 0 & \ldots & 0\\
    0 & 0 & 1 & 0 & \ldots & 0\\
    0 & 0 & 0 & 1 & \ldots & 0\\
    \vdots & \vdots & \vdots & \vdots & \vdots & \vdots\\
    0 & 0 & 0 & 0 & \ldots & 1\\
    1 & 0 & 0 & 0 & \ldots & 0\\
\end{array}
\right).
$$
 satisfies $[c,c^\dagger]_q=\1_L$ in $\Hil_L$, for all finite integer $L$. However, this situation is not so interesting for us, since the matrices for different values of $q$ are simply proportional one to the other. In this sense, this particular solution is trivial.
\end{remark}

We assume here, in view of Remark~\ref{remark1} and of the  expression (\ref{a2}) for $c$, that $C$ admits the representation
\be
\label{26}
C=\left(
\begin{array}{ccccccccc}
	0 & \beta_0 & 0 & 0 & 0 & \cdots & \cdots & 0 \\
	0 & 0 & \beta_1 & 0 & 0 & \cdots & \cdots & 0 \\
	0 & 0 & 0 & \beta_2 & 0 & \cdots & \cdots & 0 \\
	0 & 0 & 0 & 0 & \beta_3 & \cdots & . & 0 \\
	\vdots & \vdots & \vdots & \vdots & \vdots & \vdots & \vdots & \vdots \\
	0 & 0 & 0 & 0 & 0 & 0 & \cdots & \beta_{L-1} \\
	0 & 0 & 0 & 0 & \cdots &\cdots &\cdots & 0 \\
\end{array}
\right),
\en
where $\beta_n$ here will be later compared with those given in (\ref{a5}); then, we look for the consequences of this ansatz. Before starting our analysis, it might be useful to stress that, for quons, the commutation rule $[N_a,a]=-a$ is replaced by a slightly different expression,
\be
[N_c,c]=-c+(1-q)N_cc,
\label{27}
\en
where we have introduced the quonic number-like operator $N_c=c^\dagger c$. Notice that the one here is the standard commutator, while, as clearly stated before, $[.,.]_q$ is the $q$-mutator, which reduces to $[.,.]$ only when $q=1$, and to the anti-commutator $\{.,.\}$ when $q=-1$. For the truncated version $C$, we will be interested to consider the following two questions:
\begin{enumerate}
\item   Find, if possible, conditions on $\beta_n$ which produce, in analogy with (\ref{27}),
\be
[N_C,C]=-C+(1-q)N_CC,
\label{28}
\en
where $N_C=C^\dagger C$. Once these expressions of $\beta_n$ are found, if possible, we will use (\ref{25}) to deduce the expression for $K$.	
In other words: is it possible to recover for the truncated quons the same commutation rule with the number-like operator as for ordinary quons?
\item Find, if possible, conditions for (\ref{28}) to reduce to a {\em boson-like} rule
\be
[N_C,C]=-C,
\label{29}
\en
while keeping $q\neq1$. This constraint on $q$ makes the question non trivial. Indeed, if $q=1$ formula (\ref{29}) is exactly what we do expect, and there is not much to say. Different, and possibly more interesting, is the situation when $q\neq1$.
\end{enumerate}

We will show in the rest of this Section that the answers to both these questions are affirmative. Let us first consider  condition (\ref{28}), with $C$ as in (\ref{26}). Restricting to $\beta_k\neq0$ for all $k$, it is easy to check that $\beta_k^2=1+q+\cdots q^k$. Notice that $\displaystyle \sum_{j=0}^kq^j\geq0$ for all allowed values of $q$. This is indeed not different from what we have found in (\ref{a5}): the expression of $\beta_k$ remains the same both in $\Hil$ and in $\Hil_L$ (and this is why we have used the same notation in both cases).

\begin{remark} 
It may be interesting to stress that other solutions for $\beta_k$ also exist such that (\ref{28}) is satisfied. In particular we could have $\beta_0=\beta_1=\ldots=\beta_{j_0}=0$, for some $j_0<L-1$. In this case $C$ annihilates a set of orthogonal vectors, and therefore $\dim(\ker(C))>1$. We will not consider this other possibility further in this paper.
\end{remark}

From now on, to fix the ideas, we only consider  $\beta_k>0$: if $\beta_k=\sqrt{\sum_{j=0}^kq^j}$, we can deduce the expression of $K$ inverting (\ref{25}). We get 
\be
\label{210}
K=\left(
\begin{array}{ccccc}
0 & 0 & \ldots & 0 &  0\\
0 & 0 & \ldots & 0 &  0\\
\vdots & \vdots & \ddots & \vdots &\vdots\\
0 & 0 & \ldots & 0 & 0\\
0 & 0 & \ldots & 0 & \frac{\beta^2_{L}}{L+1}
\end{array}
\right),
\en 
which, as we can see, has a single non vanishing entry in the $(L+1,L+1)$ position. Incidentally we observe that $KC=0$ and that $K=K^\dagger$. However, $K=K^2$ is possible only if $q=1$. Otherwise $K\neq K^2$. This is similar, but not identical, to what we have found for the truncated bosons, and for $K_0$.

It is interesting to look at the same problem from a more abstract point of view, \emph{i.e.}, without assuming that $C$ has the form in (\ref{26}), with $\beta_k$ as above. In particular, if we assume that $C$ obeys (\ref{25}), and that $KC=0$, it is easy to check that $[N_C,C]$ is the one in (\ref{28}), independently of the explicit representation of $C$.

We can rewrite the operators $C$ and $K$ in terms of rank one operators as follows:
\be
\label{211}
K=\frac{\beta^2_{L}}{L+1}|e_{L}\rangle \langle e_{L}|\qquad\hbox{and}\qquad C=\sum_{n=0}^{L-1}\,\beta_n\,|e_{n}\rangle \langle e_{n+1}|.
\en
Here $e_n$ is the $n$-th vector of $\F_e(L)\subset\F_e$, the same sets we have introduced for bosons and truncated bosons.
We observe that, for $q=1$, $K$ coincides with $K_0$ and the operator $C$ becomes the analogous operator $A$ for the truncated bosons. If $q=-1$ and $L=1$ we recover fermionic operators. If $L\geq2$ and $q=-1$ we get extended fermionic operators, similar to those considered in \cite{bag2013}: in particular, the expression above for $K$ shows that $K=0$ for odd $L$, while $K\neq0$ for even $L$. 

Formula (\ref{26}) easily produces the following lowering and raising equations for the operators $C$ and  $C^\dagger$, respectively:
\be
\left\{\begin{array}{ll}
Ce_{n}=\beta_{n-1}\,e_{n-1}, \   \,\,\,\quad 0\leq n\leq L, \qquad  \beta_{-1}=0,\\
\\
C^\dagger e_{n}=\beta_n\,e_{n+1}, \qquad 0\leq n\leq L-1,
\end{array}\right.
\label{212}
\en
which can be extended to all of $\Hil_L$ by linearity. This is not a problem, here, since $L<\infty$, and $C$ and $C^\dagger$ are obviously bounded. From (\ref{212}) we also deduce the  following eigenvalue equations for $N_C$: $N_C\,e_n=\beta^2_{n-1}\,e_n$ for $n=0,1,2,\ldots,L$.

The vectors $e_n$, $n=1,2,\ldots,L$, are related to $e_0$ as follows:
\be
\label{213}
e_n:=\frac{(C^\dagger)^n}{\beta_{n-1}!}e_0,\quad n=0,\dots,L.
\en
Hence, $C^\dagger$ is a raising operator, satisfying $(C^\dagger)^{L+1}=0$.  The main difference with respect to the truncated bosons relies, as it is clear, in the normalization factor $\beta_{n-1}!$, which reduces to $\sqrt{n!}$ when $q=1$. We refer to the Appendix for the meaning of  $\beta_{n-1}!$, known sometimes as the $q$-factorial.

Next we go to formula (\ref{29}). We recall that what is interesting for us here is to check if (\ref{29}) is compatible with (\ref{25}) for $q\neq1$, and which are the properties of $K$ in this case. The approach is the same as above. In this case, using the expression in (\ref{26}) for $C$ in (\ref{29}), we find that the solution for $C$, restricting as before to $\beta_j>0$, reads
\be
\label{214}
C=\left(
\begin{array}{cccccc}
0 & 1 &0&0&\ldots&0\\
0&0&\sqrt{2}&0&\ldots&0 \\
0&0&0&\sqrt{3}&\ldots&0\\
\vdots&\vdots&\vdots&\vdots& \ddots&\vdots\\
0&0&0&0&\ldots&\sqrt{L}\\
0&0&0&0&\ldots&0
\end{array}
\right),
\en
which is exactly the matrix of the truncated bosons $A$ we have seen before. This does not fix $q=1$ necessarily. Indeed, inserting now this solution in (\ref{25}), we find that $K$ must be 
\be
\label{215}
K=\left(
\begin{array}{cccccc}
0 &0 &0&\ldots&0&0\\
0 & \frac{q-1}{L+1} &0&\cdots&0&0\\
0 & 0 &\frac{2(q-1)}{L+1}&\cdots&0&0\\
\vdots&\vdots&\vdots&\vdots& \vdots&\vdots\\
0&0&0&\cdots&\frac{(L-1)(q-1)}{L+1}&0\\
0&0&0& \cdots&0&\frac{1+Lq}{L+1}\\
\end{array}
\right),
\en
which is still diagonal but quite different from what we have found before, and which can be rewritten as
\be
\label{216}
K=\frac{q-1}{L+1}\sum_{n=0}^{L-1}n\,|e_{n}\rangle \langle  e_{n}|+\frac{1+qL}{L+1}\,|e_{L}\rangle \langle e_{L} |.
\en

It is particularly interesting to notice that $KC\neq0$. In fact, it turns out that  $\displaystyle KC=\frac{q-1}{L+1}N_CC$. This result can be somehow reversed: if we assume that $C$ obeys (\ref{25}) and that $\displaystyle KC=\frac{q-1}{L+1}N_CC$, then it is easy to verify that $[N_C,C]$ is the one in (\ref{29}), and the operators $C$ and $K$ have the expressions (\ref{214}) and (\ref{215}), respectively.

Again, we have lowering and raising properties for the operators $C$ and $C^\dagger$, but with purely bosonic weights:
\be
\left\{\begin{array}{ll}
Ce_{n}=\sqrt{n}\,e_{n-1}, \qquad 1\leq n\leq L, \qquad  Ce_{0}=0,\\
\\
C^\dagger e_{n}=\sqrt{n+1}\,e_{n+1}, \qquad 0\leq n\leq L-1,
\end{array}\right.
\label{217}
\en
and we recover the eigenvalue equation $N_C\,e_n=n\,e_n$ for $n=0,1,2,\ldots,L$. Also in this case, the vectors $e_n$ can be constructed out of $e_0$ as follows: 
\be
\label{218}
e_n:=\frac{(C^\dagger)^n}{\sqrt{n!}}e_0,\quad n=0,\dots,L.
\en

In conclusion,  we are back to truncated bosons, but these truncated bosons obey also a truncated 
$q$-mutation relation, see (\ref{25}), if $K$ is chosen as in (\ref{215}) and (\ref{216}). 

\subsection{Deformed TQs}\label{sect3}

In view of some older results in the literature (see \cite{bagbook}, and references therein), it might be interesting to sketch briefly if and how it is possible to deform TQs in order to obtain {\em truncated pseudo-quons} (TPQs), where the ladder operators $C$ and $C^\dagger$ are replaced by a pair $D$ and $G$ which are still ladder operators living in $\Hil_L$, but such that $G^\dagger\neq D$. Hereafter, we will propose a very simple (but still non trivial) procedure, similar to what has been done in the literature for \emph{regular} pseudo-bosons \cite{bagbook}. The main ingredients are the operators
 $C$ and $C^\dagger$, satisfying (\ref{25}), and where (\ref{26}) is the matrix form of $C$. Let  $R$ be an invertible operator on $\Hil_L$. We use $R$ and $R^{-1}$ to define $D$ and $G$ as follows:
\be
\label{31}
D=RCR^{-1}, \qquad G=RC^\dagger R^{-1}.
\en
It is clear that, if $R$ is not unitary, which is the case we will only consider here, $D\neq G^\dagger$. Using (\ref{25}) we see that $D$ and $G$ satisfy the following q-mutator: 
\be
[D,G]_q=R[C,C^\dagger]_qR^{-1}=R\left(\1_L-(L+1)K\right)R^{-1}=\1_L-(L+1)Q,
\label{32}
\en where $Q=RKR^{-1}$. It is clear that $Q^2=Q$ if $K^2=K$. This is not true, however, neither for the choice (\ref{210}) of $K$, nor for (\ref{215}), except when $q=1$.  Also, in general, $Q^\dagger\neq Q$. However, $Q^\dagger=Q$ if  $[K,R^\dagger R]=0$, since $K=K^\dagger$, at least for $K$ given in (\ref{210}) and (\ref{215}). 

Moreover, whenever $KC=0$ and $K=K^\dagger$ (which is what we have deduced out of (\ref{28})), we have $QD=RKCR^{-1}=0$, so that $D^\dagger Q^\dagger=0$. The same reasoning  implies that $Q^\dagger G^\dagger=0$, so that $GQ=0$ as well. 
	
Similarly, whenever $\displaystyle KC=\frac{q-1}{L+1}N_CC$ and $K=K^\dagger$ (which is what we have deduced from (\ref{29})),  introducing $\tilde{N}=GD$  we deduce that $\displaystyle QD=\frac{q-1}{L+1}\tilde{N}D$, implying that $\displaystyle D^\dagger Q^\dagger=\frac{q-1}{L+1}D^\dagger\tilde{N}^\dagger$,  and also $\displaystyle Q^\dagger G^\dagger=\frac{q-1}{L+1}\tilde{N}^\dagger G^\dagger$, leading to $\displaystyle GQ=\frac{q-1}{L+1}G\tilde{N}$.

\vspace{2mm}

\begin{example}
Here, we will show how the above operators look like, starting from fixed (but arbitrarily chosen) biorthogonal sets $\F_\varphi=\{\varphi_n,\, n=0,1,2,3\}$ and $\F_\psi=\{\psi_n,\,n=0,1,2,3\}$ for the Hilbert space  $\Hil_3=\mathbb{C}^4$, and using these sets as the main ingredient of all our construction. Let us consider
\be
\label{33b}
\qquad\varphi_0=\left(
\begin{array}{c}
    1\\ 0\\0\\0
\end{array}\right),
\quad
\varphi_1=\left(
\begin{array}{c}
    1\\ 1\\0\\0
\end{array}\right),\quad
\varphi_2=\left(
\begin{array}{c}
    0\\ 1\\1\\0
\end{array}\right),\quad
\varphi_3=\left(
\begin{array}{c}
    0\\ 0\\1\\1
\end{array}
\right),
\en
and
\be
\label{34b}
\qquad \psi_0=
\left(
\begin{array}{c}    
1\\ -1\\1\\-1
\end{array}\right),\quad
\psi_1=\left(
\begin{array}{c}    
 0\\ 1\\-1\\1
\end{array}\right),\quad
\psi_2=\left(
\begin{array}{c}    
    0\\ 0\\1\\-1
\end{array}\right),\quad
\psi_3=\left(
\begin{array}{c}    
    0\\ 0\\0\\1
\end{array}\right).
\en
It is clear that, in particular,  $\langle\varphi_n,\psi_m\rangle=\delta_{n,m}$. This implies that the four vectors in $\F_\varphi$ are linearly independent, and therefore that $\F_\varphi$ is a basis for $\Hil_3$. The same is true for $\F_\psi$. Hence these two sets solve the identity: $\displaystyle \sum_{n=0}^3|\psi_n\rangle\langle\varphi_n|=\sum_{n=0}^3|\varphi_n\rangle\langle\psi_n|=\1_3.$

Let us now fix $\beta_n\in\mathbb{R}, n=0,1,2$, and  define 
\be
\label{opef}
E=\sum_{n=0}^2\beta_{n}|\varphi_n\rangle\langle\psi_{n+1}|,\qquad 
 F=\sum_{n=0}^2\beta_{n}|\varphi_{n+1}\rangle\langle\psi_{n}|.
 \en
Given $f\in\mathbb{C}^4$,  we have
\be\label{effe}
EFf=\sum_{n=0}^2\beta_{n}^2\langle\psi_{n},f\rangle \varphi_{n},
\qquad
 FEf=\sum_{k=0}^2\beta_k^2\langle\psi_{k+1},f\rangle\varphi_{k+1},
 \en
so that, after some straightforward computations,
\be\label{commqef}
    [E,F]_qf=EFf-qFEf=\1_3f-4Qf,
\en
 where we have introduced the operator $Q$, say 
 $$
 Q=\frac{1}{4}\left((1-\beta_0^2)|\varphi_0\rangle\langle\psi_0|+\sum_{l=1}^2(1-\beta_l^2+q\beta_{l-1}^2)|\varphi_l\rangle\langle\psi_l|+(1+q\beta_{2}^2)|\varphi_3\rangle\langle\psi_3|\right),
 $$
 or, in matrix form,
 $$
 Q=\frac{1}{4}\left(
 \begin{array}{cccc}
  	1 -\beta_0^2 & (q+1)\beta_0^2 -\beta_1^2 & \beta_1^2-(q+1)\beta_0^2 & (q+1)\beta_0^2 -\beta_1^2\\
  	0 & 1+q\beta_0^2 -\beta_1^2 & (q+1)\beta_1^2-q\beta_0^2-\beta_2^2 & q\beta_0^2-(q+1)\beta_1^2+\beta_2^2\\
  	0 & 0 & 1+q\beta_1^2-\beta_2^2 & (q+1)\beta_2^2 -q\beta_1^2\\
  	0 & 0 & 0 & 1+q\beta_2^2
  \end{array}\right).
 $$
The operators $E$ and $F$ satisfy  the commutation rule in (\ref{32}), $[E,F]_q=\1_3-4Q$. Next we want to show that $E$ and $F$ are related to an operator $C$ (and its adjoint), obeying (\ref{25}), as in (\ref{31}). For that, we first observe that $\F_\varphi$ and $\F_\psi$ are Riesz bases. Indeed, we have $\varphi_n=R\,e_n$ and $\psi_n=\left(R^{-1}\right)^\dagger\,e_n$, where $\F_e=\{e_k\}$ is the canonical o.n. basis in $\Hil_3$, and $R$ is the invertible operator 
\be
\label{opr}
R=\left(
\begin{array}{cccc}
    1 &  1 & 0 & 0\\
    0 & 1 &  1 &0\\
    0 & 0 & 1 &  1\\
    0 & 0&0 &1
\end{array}\right).
\en

We can now determine the operator $K$ inverting the original formula for $Q$, $Q=RKR^{-1}$:
\be
\label{opk}
K=R^{-1}QR=\frac{1}{4}
\left(
\begin{array}{cccc}
        1-\beta_0^2 & 0 & 0 &  0\\
        0 & 1+q\beta_0^2-\beta_1^2 & 0 & 0\\
        0 & 0 & 1+q\beta_1^2-\beta_2^2 & 0\\
        0 & 0 & 0 & 1+q\beta_2^2
\end{array}\right).
\en
Now, since
\be\label{opc}
C=\sum_{n=0}^2\beta_n|e_n\rangle\langle e_{n+1}|,
\en
we can rewrite (\ref{opk}) as $K=\frac{1}{4}\left(\1_3-CC^\dagger+qC^\dagger C\right)$, which is exactly formula (\ref{25}). Notice also that $E=RCR^{-1}$ and $F=RC^\dagger R^{-1}$.
\end{example}

\vspace{2mm}

Going back now to our definition of $Q$, $Q=RKR^{-1}$, in the easiest case in which $[K,R]=0$ then it follows that $Q=K$. Hence,
\be
\label{32b}
[D,G]_q=\1_L-(L+1)K, 
\en
which also implies that $[G^\dagger,D^\dagger]_q=\1_L-(L+1)K$.  We are interested in finding conditions such that the analogous of (\ref{28}) and (\ref{29}) hold true. In other words, we want to understand here when the following commutators are satisfied:
\be
\label{33}
[\tilde{N},D]=-D+(1-q)\tilde{N}D,\qquad [\tilde{N}^\dagger,G^\dagger]=-G^\dagger+(1-q)\tilde{N}^\dagger G^\dagger,
\en
 and
\be
\label{34}
[\tilde{N},D]=-D,\qquad [\tilde{N}^\dagger,G^\dagger]=-G^\dagger.
\en
The equations in (\ref{33}) are satisfied if $KD=0$ and $KG^\dagger=0$, which are both true since $KC=0$. For $K$ as in (\ref{210}), it is clear that $R$ must admit the representation
\be
R=\left(
\begin{array}{cccccc}
    R_{1,1} & R_{1,2} & R_{1,3} & \ldots &R_{1,L}& 0\\
    R_{2,1} & R_{2,2} & R_{2,3} & \ldots & R_{2,L}& 0\\
    R_{3,1} & R_{3,2} & R_{3,3} & \ldots & R_{3,L}&0\\
    \vdots &\vdots & \vdots  &\ddots & \vdots&\vdots\\
    R_{L,1} & R_{L,2} & R_{L,3} & \ldots & R_{L,L}&0\\
    0 & 0& 0 &\ldots &0 & R_{L+1,L+1}
\end{array}\right),
\en
with $\det (R)\neq0$, as $R^{-1}$ must exist.

If we rather consider (\ref{34}), then a solution is given by imposing that  $\displaystyle KD=\frac{q-1}{L+1}\tilde{N}D$ and $\displaystyle KG^\dagger=\frac{q-1}{L+1}\tilde{N}^\dagger G^\dagger$, which are true since $\displaystyle KC=\frac{q-1}{L+1}N_CC$. 

In this case, $K$ has the expression in (\ref{215}), and $R$ must have the form
\be
R=\left(
\begin{array}{cccccc}
    R_{1,1} & 0 & 0 & \ldots &0& 0\\
    0 & R_{2,2} & & \ldots & 0& 0\\
    0 & 0 & R_{3,3} & \ldots & 0&0\\
    \vdots &\vdots & \vdots  &\ddots & \vdots&\vdots\\
    0 & 0 & 0 & \ldots & R_{L,L}&0\\
    0 & 0& 0 &\ldots &0 & R_{L+1,L+1}
\end{array}\right),
\en
with all the elements $R_{j,j}\neq0$, in order to have $\det(R)\neq0$.

\section{Ladders for a closed chain}\label{sect4}

In this Section, we will construct and study ladder operators which appear to be very different from those considered above. In particular, we focus on a closed chain, \emph{i.e.}, an $(N+1)$-levels system for which a ladder (or shift) operator can be defined mapping the $j$-th to the $(j+1)$-th level, $j=0,1,2,\ldots,N-1$, and mapping the $N$-th back into the zero level. For concreteness, we will choose $N=3$ here, so that $\Hil_3=\mathbb{C}^4$. The generalization to higher values of $N$ is straightforward, although the formulas become rapidly longer.

Let $\F_e=\{e_j, \,j=0,1,2,3\}$ be the canonical o.n. basis for $\Hil_3$, $\langle e_j,e_k\rangle=\delta_{j,k}$, and let $\{\gamma_j, \,j=0,1,2,3\}$ be a set of real and positive numbers. Let us consider an operator $a^\dagger$ on $\Hil_3$ acting on $\F_e$ as follows:
\be
a^\dagger e_0=\gamma_1e_1, \quad a^\dagger e_1=\gamma_2e_2, \quad a^\dagger e_2=\gamma_3e_3, \quad a^\dagger e_3=\gamma_0e_0.
\label{cc1}
\en 
Its matrix expression is 
\be
\label{cc2}
 a^\dagger=\left(
\begin{array}{cccc}
0 & 0 &0&\gamma_0\\
\gamma_1&0&0&0 \\
0&\gamma_2&0&0\\
0&0&\gamma_3&0
\end{array}\right) \quad \hbox{so that} \quad
a=
\left(
\begin{array}{cccc}
0 & \gamma_1 &0&0\\
0&0&\gamma_2&0 \\
0&0&0&\gamma_3\\
\gamma_0&0&0&0
\end{array}\right).
\end{equation}
Therefore
\be
a e_0=\gamma_0e_3, \quad a e_1=\gamma_1e_0, \quad a e_2=\gamma_2e_1, \quad a e_3=\gamma_3e_2.
\label{cc3}
\en 
The reader should be aware that the operator $a$ here has nothing to do (except for its ladder nature) with the bosonic operator $a$ introduced in Section \ref{sect2}, and that's why we decided to keep this notation.
We see that, with this choice, no vector $e_j$ is annihilated by either $a$ or $a^\dagger$. This is different from what happens for bosons and quons, truncated or not, since they all have a vacuum (for instance, see (\ref{22})), or for fermions, where both the zero and the excited levels are annihilated by $b$ and $b^\dagger$, respectively, where $\{b,b^\dagger\}=\1$. Our interest in the operators $a$ and $a^\dagger$  in (\ref{cc2}) is based on the following aspects. First of all, we are interested in describing a system with no vacuum, which appears as a finite dimensional version of what has been recently considered in \cite{bagaop2024} in connection with graphene and coherent states, or in \cite{kow} for a quantum particle on a circle. Secondly, we want to check if it is possible to construct some sort of coherent state for the operator $a$, even in our finite-dimensional settings. This is well known to be impossible as long as a vacuum exists, since an equality $a\Phi(z)=z\Phi(z)$ could never be satisfied by any non-zero $\Phi(z)$ and for all $z\in\mathbb{C}$. The reason is simple\footnote{We sketch the proof in our $\Hil_3=\mathbb{C}^4$. It is trivial to generalize it to any dimension $N<\infty$.}. As $\F_e$ is an o.n. basis for $\Hil_3$, it is clear that we can expand $\Phi(z)$ as follows: $\Phi(z)=\sum_{k=0}^3c_k(z)e_k$, where $c_k(z)$ are $z$-dependent coefficients. If $A$ is a generic lowering operator for $\F_e$, and if $Ae_0=0$, it is clear that $A\Phi(z)$ is a linear combination of $e_0,e_1$ and $e_2$, but not of $e_3$. On the other hand, $\displaystyle z\Phi(z)=\sum_{k=0}^3z\,c_k(z)e_k$. It is easy to check that the only solution,  
independently of the details of action of the lowering operator $A$, is that $c_k(z)=0$ for all $k$. Hence, 
$\Phi(z)=0$, which is not allowed\footnote{The obvious (and well known) way out is to work in an infinite dimensional Hilbert space. In this case, a non trivial solution can be usually found. This is the case, just to cite an example, of the bosonic lowering operator.}. Last but not least, we would like to consider new operators that can be eventually used in some dynamical system, on the same lines it has been done in recent years in many applications where ladder operators of different nature appear in the analysis of macroscopic systems, \cite{bagbook1,bagbook2,bagbook3}. 

Going back to the matrices $a$ and $a^\dagger$ in (\ref{cc2}), we see that
\be
N=a^\dagger a=\left(
\begin{array}{cccc}
	\gamma_0^2 & 0 &0&0\\
	0&\gamma_1^2&0&0 \\
	0&0&\gamma_2^2&0\\
	0&0&0&\gamma_3^2
\end{array}\right),
\label{cc4}
\en
which is diagonal despite of the presence of the $\gamma_0$ in $a$ and $a^\dagger$. The commutator between these two operators is also diagonal,
\be
[a,a^\dagger]=\left(
\begin{array}{cccc}
	\gamma_1^2-\gamma_0^2 & 0 &0&0\\
	0&\gamma_2^2-\gamma_1^2&0&0 \\
	0&0&\gamma_3^2-\gamma_2^2&0\\
	0&0&0&\gamma_0^2-\gamma_3^2
\end{array}\right)=:\Gamma
\label{cc5}
\en
Incidentally, we observe that $[a,a^\dagger]\neq\1$ for any choice of the $\gamma_j$'s. This is not a surprise, as it is in agreement with well known facts for bosonic operators: no operator $X$ can satisfy $[X,X^\dagger]=\1$ in any finite dimensional Hilbert space. However, these operators share some interesting features with ordinary bosons. Indeed, we have $[N,a]=-\Gamma\,a$ and $[N,a^\dagger]=a^\dagger\Gamma$. We also see that $[N,\Gamma]=[M,\Gamma]=0$, where $M=aa^\dagger=\Gamma+N$. These commutators are similar to the analogous results we get for bosons, as we can check by replacing $\1$ with $\Gamma$. In particular, if we call $a(t)=\exp(iNt)a\exp(-iNt)$ (as a sort of Heisenberg dynamics for $a$, driven by the ``Hamiltonian'' $H=N$), it is easy to see that
$$
\frac{da(t)}{dt}=i\exp(iNt)[N,a]\exp(-iNt)=-i\Gamma a(t) \qquad\Rightarrow\qquad a(t)=\exp(-i\Gamma t) a,
$$
as $\Gamma(t)=\exp(iNt)\Gamma \exp(-iNt)=\Gamma$. This implies that $a^\dagger(t)a(t)=a^\dagger a$, as it is obvious because $[H,N]=0$, if $H=N$.

A simple way to see that our ladder operator $a$ describes a closed chain consists in observing that
\be
(a^\dagger)^4=\gamma_0\gamma_1\gamma_2\gamma_3\1=a^4,
\label{cc6}\en
which implies that the fourth power of both $a$ and $a^\dagger$ are simply proportional to the identity operator: we have just gone all over the chain, going back to the initial state. 

\begin{remark} It is easy to understand that the operator $a^\dagger$ can be expressed in terms of a family $(b_j,b_j^\dagger )$ of fermionic operators as follows: let $b_j$ be such that $\{b_j,b_k^\dagger\}=\delta_{j,k}\1$, and $b_j^2=0$, $j=0,1,2,3$, with all the other anti-commutators being zero. The operator $b_j$ annihilates a particle in the $j$-th level, while $b_j^\dagger$ creates such a particle. Hence, we can rewrite
$$
a^\dagger=\sum_{j=0}^3 b_jb_{j+1}^\dagger,
$$ 
with the identification $b_4^\dagger\equiv b_0^\dagger$; $a^\dagger$ and $\sum_{j=0}^3 b_jb_{j+1}^\dagger$ behave effectively in the same way. Anyway, it is clear that the use of $a$ is much more efficient than adopting the $b_j$ operators.
\end{remark}

\begin{remark} 
Another possible choice for $a$ in (\ref{cc2}) satisfying  (\ref{cc6}) was already considered in the Remark~\ref{remark1}. In fact, setting 
$$
a=\frac{1}{\sqrt{1-q}}\left(
\begin{array}{llll}
    0 & 1 & 0 & 0\\
    0 & 0 & 1 & 0\\
    0 & 0 & 0 & 1\\
    1 & 0 & 0 & 0
\end{array}\right),
$$
$a$ satisfies (\ref{cc6}) if $q$ is such that  $\gamma_0\gamma_1\gamma_2\gamma_3=\frac{1}{(1-q)^2}$. Notice that this is surely possible whenever $\gamma_0\gamma_1\gamma_2\gamma_3\geq\frac{1}{4}$. 
\end{remark}

It is clear that all the results we have deduced here can be easily generalized to $\Hil_L$ for $L\neq3$. Notice that, here, $L+1$ is just the length of the closed chain we are interested to.

\subsection{Discrete coherent states for $a$}

We have already shown that, for a purely lowering operator in a finite dimensional vector space, it is not possible to find a vector which could be called {\em a coherent state}. Hereafter, we propose a general proof of this statement, together with some related results. In the last part of this Section, we go back to our ladder operators $a$ and $a^\dagger$, and analyze in some details what Theorem \ref{thm1} below implies.

Let $A$ be a general $n\times n$ matrix (no other requirement: in particular, we are not assuming $A$ is a lowering operator as we did before). Our general question is the following: is it possible to find a $z$-dependent non-zero vector $\Phi(z)\in\Hil$, $\Hil=\mathbb{C}^n$, such that $A\Phi(z)=z\Phi(z)$ for all $z\in\mathbb{C}$? Of course, we cannot expect the answer be positive in general, as it is well known that coherent states for fermionic operators ($n=2$ and $A=\sigma^+$, one of the Pauli matrices) can only be defined if $z$ is replaced by some Grassmann number \cite{gazeaubook}. Indeed, this is a general aspect, which does not depend on the specific form of the matrix $A$.

\begin{thm}\label{thm1}
	The only solution of the identity $A\Phi(z)=z\Phi(z)$, $\forall z\in\mathbb{C}$, is $\Phi(z)=0$. 
However, it is possible to find $z_1,z_2,\ldots,z_n$, not necessarily different, and $n$ non-zero vectors $\Phi(z_j)$, $j=1,2,\ldots,n$, such that
\be
A\Phi(z_j)=z_j\Phi(z_j), 
\label{sc1}\en
$j=1,2,\ldots,n$. Moreover, if these $z_j$ are all different, $\F_\Phi=\{\Phi(z_j), \,j=1,2,\ldots,n\}$ is a basis for $\Hil$. Its biorthonormal basis, $\F_\psi=\{\psi(z_j), \,j=1,2,\ldots,n\}$, is a set of eigenstates of $A^\dagger$ with eigenvalues $\overline{z_j}$:
\be
A^\dagger\psi(z_j)=\overline{z_j}\psi(z_j), 
\label{sc2}\en
$j=1,2,\ldots,n$. 

\end{thm} 

\begin{proof}
First we observe that, rewriting $A\Phi(z)=z\Phi(z)$ as $(A-z\1)\Phi(z)=0$, this admits a non-trivial solution only if $p(z)=\det(A-z\1)=0$. But $p(z)$ is a polynomial in $z$ of degree $n$. Hence, it has exactly $n$ roots, $z_1,z_2\ldots,z_n$ not necessarily all different. For all $\tilde z\in\mathbb{C}$ not in $\sigma(A)=\{z_1,z_2\ldots,z_n\}$, the set of eigenvalues of $A$, $p(\tilde z)\neq0$, and the only solution of $A\Phi(\tilde z)=\tilde z\Phi(\tilde z)$ is $\Phi(\tilde z)=0$: hence, $A\Phi(z)=z\Phi(z)$ does not admit a non trivial solution $\Phi(z)$ for all $z\in\mathbb{C}$. 

Now, let us assume that all the eigenvalues $z_j$ of $A$ are distinct. Since we are not assuming that $A=A^\dagger$, each $z_j$ is in general complex. If they are all different, the $n$ eigenvectors in $\F_\Phi$ are all linearly independent. Hence, they form a basis for $\Hil$. Thus \cite{chri,heil}, there exists a unique biorthonormal basis  $\F_\psi=\{\psi(z_j), \,j=1,2,\ldots,n\}$. Then we have $\langle\Phi(z_j),\psi(z_k)\rangle=\delta_{j,k}$, and therefore
\begin{eqnarray*}
 \langle \Phi(z_j),A^\dagger\psi(z_k)\rangle&=&\langle A\Phi(z_j),\psi(z_k)\rangle= \overline{z_j}\langle\Phi(z_j),\psi(z_k)\rangle=\overline{z_j}\delta_{j,k}=\\
&=&\overline{z_k}\delta_{j,k}=\langle \Phi(z_j),\overline{z_k}\psi(z_k)\rangle, \qquad \forall j,k=1,2,\ldots,n.
\end{eqnarray*}
This implies, in particular, that $\langle \Phi(z_j),A^\dagger\psi(z_k)-\overline{z_k}\psi(z_k)\rangle=0$, for all $j$. Formula (\ref{sc2}) now follows from the fact that $\F_\Phi$ is complete.

\end{proof}

From this theorem we conclude that each vector $f\in\Hil$ can be written as
\be
f=\sum_{j=1}^n\langle\Phi(z_j),f\rangle \psi(z_j)=\sum_{j=1}^n\langle\psi(z_j),f\rangle \Phi(z_j),
\label{sc3}\en
or, equivalently, that
\be
\sum_{j=1}^n|\psi(z_j)\rangle\langle \Phi(z_j)|=\sum_{j=1}^n|\Phi(z_j)\rangle\langle \psi(z_j)|=\1.
\label{sc4}\en
The interpretation is that, under the assumption of the theorem and independently of the fact that $A$ is a ladder operator or not, a sort of \emph{discrete set of coherent states} can be found, and this set, together with its biorthogonal family, resolves the identity in $\Hil$.

\begin{remark} It might appear not so natural to call these vectors coherent states, since these are states with very specific properties \cite{gazeaubook}. However, we like to keep this name, in view of the application to the chain, and to the fact that also \emph{ordinary} coherent states obey a discrete resolution of the identity \cite{zak}.
\end{remark}

Let us now go back to the operators $a$ and $a^\dagger$ in (\ref{cc2}).  The eigenvalues for $a$ are 
\be
\label{sc5}
E_0=-\gamma, \qquad E_1=-i\gamma, \qquad E_2=i\gamma, \qquad E_3=\gamma,
\en
where $\gamma=\sqrt[4]{\gamma_0\gamma_1\gamma_2\gamma_3}>0$. The corresponding eigenvectors are
\be
\label{sc6}
\begin{array}{l}
\displaystyle\Phi(z_0)=\frac{1}{2}\left(
-\frac{\gamma}{\gamma_0},
\frac{\gamma^2}{\gamma_0\gamma_1},
-\frac{\gamma_3}{\gamma},
1\right)^T, \\
\displaystyle\Phi(z_1)=\frac{1}{2}\left(
-i\frac{\gamma}{\gamma_0},
-\frac{\gamma^2}{\gamma_0\gamma_1},
i\frac{\gamma_3}{\gamma},
1\right)^T, \\ 
\displaystyle\Phi(z_2)=\frac{1}{2}\left(
	i\frac{\gamma}{\gamma_0},
-\frac{\gamma^2}{\gamma_0\gamma_1},
-i\frac{\gamma_3}{\gamma},
1
\right)^T, \\ 
\displaystyle\Phi(z_3)=\frac{1}{2}\left(
	\frac{\gamma}{\gamma_0},
\frac{\gamma^2}{\gamma_0\gamma_1},
\frac{\gamma_3}{\gamma},
1
\right)^T. 
\end{array}
\en

The biorthogonal set $\F_\psi=\{\psi(z_j)\}$ is the set of the following vectors:
\be
\label{sc7}
\begin{array}{l}
\displaystyle\psi(z_0)=\frac{1}{2}\left(
	-\frac{\gamma_0}{\gamma},
	\frac{\gamma_0\gamma_1}{\gamma^2},
	-\frac{\gamma}{\gamma_3},
	1
\right)^T, \\ 
\displaystyle\psi(z_1)=\frac{1}{2}\left(
i\frac{\gamma_0}{\gamma},
-\frac{\gamma_0\gamma_1}{\gamma^2},
-i\frac{\gamma}{\gamma_3},
1
\right)^T, \\ 
\displaystyle\psi(z_2)=\frac{1}{2}\left(
-i\frac{\gamma_0}{\gamma},
-\frac{\gamma_0\gamma_1}{\gamma^2},
i\frac{\gamma}{\gamma_3},
1
\right)^T, \\ 
\displaystyle\psi(z_3)=\frac{1}{2}\left(
\frac{\gamma_0}{\gamma},
\frac{\gamma_0\gamma_1}{\gamma^2},
\frac{\gamma}{\gamma_3},
1
\right)^T.
\end{array} 
\en
It is possible to check that $a^\dagger\psi(z_j)=\overline{E_j}\psi(z_j)$, $j=0,1,2,3$, in agreement with 
Theorem~\ref{thm1}, and that $\F_\Phi$ and $\F_\Psi$, together, resolve the identity, as in (\ref{sc4}).

\vspace{2mm}

\begin{remark}

It is not difficult to extend this construction to higher dimensions, i.e. to operators $a$ and $a^\dagger$ that, despite of formulas in (\ref{cc2}), are represented by larger $M\times M$ matrices, as far as they still represent ladders on a closed chain. Of course, they will act on some o.n. basis which is made by $M$ elements, in a way which extends the one in (\ref{cc1}).

\end{remark} 

\vspace{2mm}

We could expand these vectors in terms of the canonical orthonormal basis $\F_e=\{e_j, \,j=0,1,2,3\}$. For instance, we have 
$$
\Phi(z_0)=\frac{1}{2}\left(-\frac{\gamma}{\gamma_0}e_0+\frac{\gamma^2}{\gamma_0\gamma_1}e_1-\frac{\gamma}{\gamma_3}e_2+e_3\right),
$$ 
so that, using (\ref{cc4}),
$$
N\Phi(z_0)=\frac{1}{2}\left(-\gamma\gamma_0 e_0+\frac{\gamma^2\gamma_1}{\gamma_0}e_1-\frac{\gamma\gamma_2^2}{\gamma_3}e_2+\gamma_3^2e_3\right).
$$
This result suggests that $\Phi(z_0)$ is not stable under the time evolution produced by $N$, \emph{i.e.}, under the operator $\exp(iNt)$. However, this does not exclude that $\exp(iNt)\Phi(z_0)$ can coincide, or be proportional, to some of the other vectors in (\ref{sc6}) or in (\ref{sc7}) for some specific time $t$. This is another (big) difference with what people usually call coherent states, since these are stable under time evolution \cite{gazeaubook}.

Another difference with the {\em canonical} coherent states (as those attached to the harmonic oscillator), is that, here, both the lowering and the raising operators, $a$ and $a^\dagger$, admit eigenstates belonging to the Hilbert space where the operators act, and that these two families are biorthonormal. This is due to the fact that we are working with a simplified functional settings, \emph{i.e.}, in a finite dimensional Hilbert space. This, as already commented before, produces several interesting consequences.

\section{Conclusions}\label{sect5}

In this paper, we have considered different classes of finite-dimensional matrices obeying interesting algebraic relations, and with interesting physical interpretations. In particular, we have considered first a truncated version of quons, and discussed some of their possible deformations, possibly useful for non-Hermitian Quantum Mechanics. Then, we have introduced and analyzed ladder operators which are defined on certain closed chains of finite length. These latter operators can be seen as discrete translation operators, moving a given particle from one site of the chain to another. 

We have further considered the general problem of the existence of a sort of coherent state for our lowering operators, and we have deduced that this is possible, but in a slightly different way with respect to what it is usually done for canonical coherent states. Indeed,  we found a finite set of vectors which are eigenstates of the lowering operator $A$ and which are (under mild assumptions) a basis for $\Hil$. Its (unique) biorthonormal basis turns out to be a set of eigenstates of $A^\dagger$. Together, these two families resolve the identity.

In a forthcoming paper, now in preparation, we plan to use the operators considered here in a \emph{dynamical context}, \emph{i.e.}, in the construction and analysis of some specific dynamical systems described in terms of operators, on the same lines widely discussed in \cite{bagbook1,bagbook2,bagbook3}.

\section*{Acknowledgments}
F. B.  acknowledges partial financial support from Palermo University and by   project ICON-Q, Partenariato Esteso NQSTI - PE00000023, Spoke 2.. F. O.  acknowledges partial financial support from Messina University. F. B. and F. O.  acknowledges partial financial support from the G.N.F.M. of the INdAM and by the PRIN grant {\em Transport phenomena in low dimensional
	structures: models, simulations and theoretical aspects}- project code 2022TMW2PY - CUP B53D23009500006. 
A. F. acknowledges partial financial support from Palermo University for his visit in Messina, and all the people in Messina University for their warm welcome during his days in the MIFT Department.

\appendix

\section{Quons}
\label{appendix} 

In this appendix, we briefly review few facts on  quons, starting with their definition.

Quons are defined  by  {\it q-mutator} 
\be 
\label{a1}
[c,c^\dagger]_q:=c c^\dagger -q c^\dagger c=\1, \qquad q\in [-1,1],  
\end{equation}
where the creation and the annihilation operators are $c^\dagger$
and $c$ and of course  the CCR are obtained for $q=1$, while  CAR (canonical anticommutation relations) for $q=-1$. Note that the condition $c^2=0$ is not required here a priori, so that it should be added, when $q=-1$. For  $q$ in
the interval $(-1,1)$, we have that (\ref{a1}) describes particles which are neither
bosons nor fermions.  

In \cite{moh}, it is proved that the eigenstates of $N_0=c^\dagger\,c$ are analogous to the bosonic ones, except that for the normalization. A simple concrete realization of  (\ref{a1}) can be deduced as follows.

Let $\mathcal{F}_e=\{e_k, \, k=0,1,2,\ldots\}$ be the canonical orthonormal  basis of the Hilbert space $\Hil=l^2(\mathbb N_0)$, with all zero entries except that in the $(k+1)$th position, which is equal to one: 
\be
\langle e_k,e_m\rangle=\delta_{k,m}.
\en  
If we take
\be
c=\left(
\begin{array}{ccccccc}
0 & \beta_0 & 0 & 0 & 0 & 0 & \cdots \\
0 & 0 & \beta_1 & 0 & 0 & 0 & \cdots \\
0 & 0 & 0 & \beta_2 & 0 & 0 & \cdots \\
0 & 0 & 0 & 0 & \beta_3 & 0 & \cdots \\
0 & 0 & 0 & 0 & 0 & \beta_4  & \cdots \\
\cdots & \cdots & \cdots & \cdots & \cdots & \cdots & \cdots \\
\cdots & \cdots & \cdots & \cdots & \cdots & \cdots & \cdots \\
\end{array}
\right),
\label{a2}
\en
then (\ref{a1}) is satisfied if 
\be
\beta_0^2=1 \qquad \hbox{and} \qquad \beta_n^2=1+q\beta_{n-1}^2, \qquad  \forall n\geq1,
\en 
which implies that
\be
\label{a5}
\beta_n^2=\left\{
\begin{array}{ll}
n+1, \  \quad \mbox{if} \   \ q=1,\\
\\
\frac{1-q^{n+1}}{1-q}, \ \ \mbox{if} \ \ q\neq 1.
\end{array}
\right.
\en
	
In older papers the authors have often restricted to $\beta_n>0$ for all $n\geq0$. The above form of $c$ shows that  $c\,e_0=0$, and $c^\dagger$ behaves as a raising operator. From (\ref{a2}) we deduce
\be
e_{n+1}=\frac{1}{\beta_n}\,c^\dagger e_n=\frac{1}{\beta_n!}(c^\dagger)^{n+1}\,e_0,
\label{a3}
\en
for all $n\geq0$. Here we have introduced the $q$-\textit{factorial}: 
\be
\beta_n!=\beta_n\beta_{n-1}\cdots\beta_2\beta_1,
\en 
which is often indicated as ${[n]}_q$ in some literature. Of course, from (\ref{a3}) it follows that  \be
c^\dagger e_n=\beta_ne_{n+1}.
\en
Using (\ref{a2}), it is also easy to check that $c$ acts as a lowering operator on $\mathcal{F}_e$: \begin{equation}c\,e_m=\beta_{m-1}e_{m-1} \ \ \ \ \ \ \ \forall m \ge 0,\end{equation} where it is also useful to introduce $\beta_{-1}=0$, in order to ensure  $c\,e_0=0$. Then we have
\be
N_0e_m=\beta_{m-1}^2e_m, \qquad \forall m \ge 0.
\label{a4}
\en
The operator $N$ is formally defined in \cite{moh} as  
\be
N=\frac{1}{\log(q)}\,\log(\1-N_0(1-q)) \quad \hbox{for}  \ 0<q<1,
\en 
and satisfies the eigenvalue equation 
\be
Ne_m=me_m, \qquad \forall m\ \ge 0.
\en

It might be interesting to stress that there are other ways to represent the operator $c$ in (\ref{a2}). We refer to \cite{erem} for such an alternative representation in $\Lc^2(\mathbb{R})$.

\end{document}